\documentclass[a4paper,USenglish,10pt]{article}
\usepackage[a4paper,margin=1.1in]{geometry}
\usepackage{latexsym}
\usepackage{amsmath,amssymb,amsthm,mathtools}
\usepackage{amsfonts}
\usepackage{cite}
\usepackage{hyperref}
\usepackage[linesnumbered,lined,ruled,noend,vlined]{algorithm2e}
\usepackage[appendix=inline]{apxproof}
\usepackage{cleveref}
\usepackage{tabularx,environ}
\usepackage{comment}
\usepackage{url}
\usepackage{authblk}
\usepackage[bold,full]{complexity}
\usepackage{lscape}
\usepackage{xifthen}
\usepackage{lineno}
\usepackage{siunitx}
\usepackage{color}

\usepackage{tikz}
\usetikzlibrary{decorations.pathreplacing}

\newtheoremrep{corollary}{Corollary}

\newcommand{\set}[1]{\{#1\}}
\newcommand{\relmiddle}[1]{\mathrel{}\middle#1\mathrel{}}
\newcommand{\inset}[2]{\left\{#1 \relmiddle| #2\right\}}
\newcommand{\size}[1]{| #1|}
\newcommand{\order}[1]{O(#1)}

\newtheorem{lemma}{Lemma}
\newtheorem{theorem}[lemma]{Theorem}
\newtheorem{proposition}[lemma]{Proposition}
\newtheorem{definition}[lemma]{Definition}

\newcommand{\comp}[1]{\mu(#1)}

\makeatletter
\@ifpackageloaded{algorithm2e}{
\SetKwProg{Fn}{Function}{}{}
\SetKwProg{Procedure}{Procedure}{}{}
\SetKwProg{Subprocedure}{Subprocedure}{}{}
\SetKwComment{tcc}{//}{}
\SetKwFunction{Output}{Output}%
\SetKwFunction{AllChildren}{AllChildren}
\SetKwFunction{LexMin}{LexMin}
\SetKwFunction{Traverse}{Traverse}
\SetKwFunction{Traversetree}{Traverse-tree}
\SetKwInOut{AlgInput}{Input}
\SetKwInOut{AlgOutput}{Output}
\SetKwInOut{AlgPrecondition}{Pre-conditions}
\SetKwInOut{AlgInvariant}{Invariants}
}{}
\makeatother

\title{Polynomial-Delay Enumeration of Large Maximal Matchings}

\author[1]{Yasuaki Kobayashi}
\author[2]{Kazuhiro Kurita}
\author[3]{Kunihiro Wasa}

\affil[1]{Graduate School of Information Science and Technology, Hokkaido University, Sapporo, Japan}
\affil[2]{Graduate School of Informatics, Nagoya University, Nagoya, Japan}
\affil[3]{Department of Computer Science and Engineering, Toyohashi University of Technology, Aichi, Japan}
\date{}

\begin{document}

\maketitle

\begin{abstract}
Enumerating matchings is a classical problem in the field of enumeration algorithms.
There are polynomial-delay enumeration algorithms for several settings, such as enumerating perfect matchings, maximal matchings, and (weighted) matchings in specific orders.
In this paper, we present polynomial-delay enumeration algorithms for maximal matchings with cardinality at least given threshold $t$.
Our algorithm enumerates all such matchings in $O(nm)$ delay with exponential space, where $n$ and $m$ are the number of vertices and edges of an input graph, respectively.
We also present a polynomial-delay and polynomial-space enumeration algorithm for this problem.
As a variant of this algorithm, we give an algorithm that enumerates $k$-best maximal matchings that runs in polynomial-delay.
\end{abstract}

\section{Introduction}\label{sec:intro}

Computing a maximum cardinality matching in graphs is a fundamental problem in combinatorial optimization and has numerous applications in many theoretical and practical contexts.
This problem is well known to be solvable in polynomial time by the famous blossom algorithm due to Edmonds~\cite{Edmonds:paths:1965}.
This algorithm runs in time $O(n^2m)$, where $n$ and $m$ are the numbers of vertices and edges of an input graph, and the running time is improved to $O(n^{1/2}m)$~\cite{DBLP:conf/focs/MicaliV80} and $O(n^{\omega})$~\cite{DBLP:conf/focs/MuchaS04}, where $\omega < 2.37$ is the matrix multiplication exponent.

Enumerating matchings in graphs is also a well-studied problem in the literature~\cite{Uno:NIIjournal,DBLP:conf/isaac/Uno97,DBLP:journals/dam/ChegireddyH87,Fukuda:1992,Uno:2015}.
In this problem, we are given a graph $G = (V, E)$ and the goal is to compute all matchings of $G$ satisfying some prescribed conditions.
This is motivated by a typical situation that a single optimal matching can be inadequate for real-world problems since intricate constraints and preferences emerging in real-world problems are overly simplified or even ignored to solve the problem efficiently.
In this situation, multiple near optimal solutions are preferable rather than a single optimal solution.

There are two lines of research for enumerating matchings.
The problem of enumerating inclusion-wise maximal matchings is a special case of enumerating maximal independent sets or cliques in graphs, which is one of the most prominent problems in the field of enumeration algorithms~\cite{DBLP:journals/algorithmica/CominR18,Alessio:Roberto:ICALP:2016,Johnson:Yannakakis:Papadimitriou:IPL:1988,DBLP:conf/swat/MakinoU04,DBLP:journals/siamcomp/TsukiyamaIAS77}.
Tsukiyama et al.~\cite{DBLP:journals/siamcomp/TsukiyamaIAS77} showed that the problem of enumerating all maximal independent sets in graphs is solvable in $\order{nm}$ delay and polynomial space.
Johnson et al.~\cite{Johnson:Yannakakis:Papadimitriou:IPL:1988} also discussed a similar algorithm for this problem.
Makino and Uno~\cite{DBLP:conf/swat/MakinoU04} and Comin and Rizzi~\cite{DBLP:journals/algorithmica/CominR18} improved the running time for dense graphs via fast matrix multiplication algorithms.
For maximal matching enumeration, Uno~\cite{Uno:NIIjournal} gave an $\order{n + m + \Delta N}$-time algorithm for enumerating all maximal matchings of graphs, which substantially improves the known algorithm for enumerating maximal independent sets in general graphs~\cite{DBLP:journals/siamcomp/TsukiyamaIAS77} when input graphs are restricted to line graphs.
Here, $\Delta$ is the maximum degree and $N$ is the number of maximal matchings in an input graph.

The other line of work is to enumerate matchings with cardinality or weight constraints.
One of the best known results along this line is based on \emph{$k$-best enumeration}~\cite{DBLP:reference/algo/Eppstein16}.
Here, we say that an enumeration algorithm for (weighted) matchings is a \emph{$k$-best enumeration algorithm} if given an integer $k$, the algorithm enumerates $k$ distinct matchings $\mathcal M = \{M_1, \ldots, M_k\}$ of $G$ such that every matching in $\mathcal M$ has cardinality (or weight) not smaller than that not in $\mathcal M$.
In the 1960s, Murty developed a $k$-best enumeration algorithm based on a simple binary partition technique~\cite{Murty:Letter:1968}. 
Lawler~\cite{Lawler1972} generalized Murty's algorithm to many combinatorial problems, and then $k$-best enumeration algorithms for other problems have been discussed in various fields (see \cite{DBLP:reference/algo/Eppstein16} for a survey).
Chegireddy and Hamacher~\cite{DBLP:journals/dam/ChegireddyH87} developed an $\order{kn^3}$-time $k$-best enumeration algorithm for weighted perfect matchings in general graphs.

In this paper, we focus on enumerating matchings satisfying both maximality and cardinality conditions.
More specifically, we address the following problems.

\begin{definition}
    Given a graph $G = (V, E)$ and a non-negative integer $t$, {\sc Large Maximal Matching Enumeration} asks to enumerate all maximal matchings of $G$ with cardinality at least $t$.
\end{definition}
\begin{definition}
    Given a graph $G = (V, E)$ and a non-negative integer $k$, {\sc $k$-Best Maximal Matching Enumeration} asks to compute a set $\mathcal M$ of $k$ maximal matchings of $G$ such that the cardinality of any maximal matching in $\mathcal M$ is not smaller than that not in $\mathcal M$.
\end{definition}

These kind of problems are recently focused in several work~\cite{Kobayashi:Efficient:2020, DBLP:journals/corr/abs-2012-09153}, where they considered the problems of enumerating minimal solutions with weight or cardinality constraints.
We would like to mention that satisfying both weight or cardinality and maximal/minimal constraints makes enumeration problems even more difficult: Korhonen \cite{DBLP:journals/corr/abs-2012-09153} showed that the problem of enumerating minimal separators of cardinality at most $k$ is solvable in incremental polynomial time or FPT delay, while the problem of enumerating minimal separators without cardinality constraint $k$ is solvable in polynomial delay~\cite{Takata::2010}.

The results of our paper are as follows.
We observe that a straightforward application of the binary partition technique~\cite{Birmele:2013,Fukuda:1992} would not yield polynomial-delay algorithms for {\sc Large Maximal Matching Enumeration}: This technique is typically based on the \emph{extension problem}, which will be defined in \Cref{sec:hard}, and we prove that this problem for {\sc Large Maximal Matching Enumeration} is \NP-hard even if $t = 0$.
This result is independently shown by Casel et al.\cite{DBLP:conf/fct/CaselFGMS19}. See Theorem~10 in \cite{DBLP:conf/fct/CaselFGMS19}.
As algorithmic results, we present $\order{nm}$-delay enumeration algorithms for \textsc{Large Maximal Matching Enumeration} and \textsc{$k$-Best Maximal Matching Enumeration}. 
These algorithms run in exponential space. 
Note that for {\sc $k$-Best Maximal Matching Enumeration}, our algorithm requires $\Omega(k)$ space, while this is indeed exponential when $k$ is exponential in $n$.
We also present an $O(n^2\Delta^2)$-delay and polynomial-space enumeration algorithm for {\sc Large Maximal Matching Enumeration}.

Whereas our algorithms are slower than the known algorithm of Uno~\cite{Uno:NIIjournal}, our algorithms only enumerate matchings that are maximal and have cardinality at least $t$, which would be more efficient with respect to the overall performance when the number of ``large'' maximal matchings are sufficiently smaller than that of all maximal matchings. See \Cref{fig:many_small_matchings} for such an example.

\begin{figure}[t]
    \centering
    \begin{tikzpicture}[scale=0.5]
        \pgfdeclarelayer{nodelayer}
\pgfdeclarelayer{edgelayer}
\pgfsetlayers{nodelayer,edgelayer}

\tikzstyle{matching vertex} = [circle,fill=black,inner sep=0pt,minimum size=4pt]
\tikzstyle{dots} = [circle,fill=black,inner sep=0pt,minimum size=2pt]
\tikzstyle{matching edge} = [black]

\begin{pgfonlayer}{nodelayer}
	\node (0) at (1.5, 3.5) {};
	\node (1) at (1.5, 2) {};
	\node (2) at (7, 2) {};
	\node (3) at (7, 3.5) {};
	
	\node [style = matching vertex] (4) at (2.5, 2.75) {};
	\node [style = matching vertex] (5) at (2.5, 0.5) {};
	\node [style = matching vertex] (6) at (3.5, 2.75) {};
	\node [style = matching vertex] (7) at (3.5, 0.5) {};
	\node [style = matching vertex] (10) at (6, 2.75) {};
	\node [style = matching vertex] (11) at (6, 0.5) {};
	
	\node [style = dots] (d1) at (4.25, 1.5) {};
	\node [style = dots] (d2) at (4.75, 1.5) {};
	\node [style = dots] (d3) at (5.25, 1.5) {};
\end{pgfonlayer}

\begin{pgfonlayer}{edgelayer}
	\draw[rounded corners=7pt] (0.center) -- (3.center) -- (2.center) -- (1.center) -- cycle;
	\draw[style = matching edge] (4) to (5);
	\draw[style = matching edge] (7) to (6);
	\draw[style = matching edge] (10) to (11);
	
    \draw[thick, decoration={brace,mirror,raise=5pt},decorate]
      (5.south) -- node[below=6pt] {$2n$} (11.south);
\end{pgfonlayer}
    \end{tikzpicture}
    \caption{The graph is obtained from the complete graph $K_{2n}$ (depicted as the rounded rectangle) with $2n$ vertices by adding a pendant vertex to each vertex. Since $K_{2n}$ has $\frac{(2n)!}{2^n\cdot n!}$ perfect matchings, the graph contains exactly one maximal matchings of cardinality $2n$ and at least $\frac{(2n)!}{2^n\cdot n!}$ maximal matchings.}
    \label{fig:many_small_matchings}
    
\end{figure}
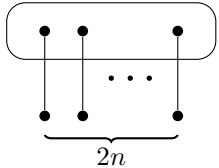


Our algorithms are based on the supergraph technique, which is frequently used in designing enumeration algorithms~\cite{Cohen::2008,Conte::2019,DBLP:conf/mfcs/KuritaK20,Schwikowski::2002,Khachiyan::2008,Khachiyan::2006:matroid}.
In this technique, we define a directed graph on the set of all solutions and the enumeration algorithm simply traverses this directed graph from an arbitrary solution.
To enumerate all solutions, we need to carefully design this directed graph so that all the nodes can be traversed from an arbitrary solution.
We basically follow the technique due to Cohen et al.~\cite{Cohen::2008}, which allows to define a suitable directed graph for enumerating maximal matchings.
We carefully analyze this directed graph and prove that this directed graph has a ``monotone'' path from an arbitrary maximum matching to any maximal matching of $G$, where we mean by a monotone path a sequence of maximal matchings $(M_1, M_2, \ldots, M_k)$ with $|M_1| \ge |M_2| \ge \cdots \ge |M_k|$.
This also enables us to enumerate all maximal matchings in a non-decreasing order of its cardinality.
Let us note that our approach is different from those in the maximal matching enumeration~\cite{Uno:NIIjournal} and the $k$-best enumeration for matchings~\cite{Murty:Letter:1968,DBLP:journals/dam/ChegireddyH87}.
Our polynomial-space enumeration algorithm also exploits this monotone path.

\section{Preliminaries}\label{sec:prelim}
Let $G = (V, E)$ be a graph.
Let $n = \size{V}$ and $m = \size{E}$.
Throughout this paper, we assume that $G$ has no self-loops and parallel edges.
We also assume that $G$ has no isolated vertices and hence we have $n = O(m)$. 
The vertex set and edge set of $G$ are denoted by $V(G)$ and $E(G)$, respectively.
For a vertex $v \in V$, the set of edges incident to $v$ is denoted by $\Gamma(v)$.
To simplify the notation, we also use $\Gamma(e)$ to denote $(\Gamma(u) \cup \Gamma(v)) \setminus \{e\}$ for each edge $e = \{u, v\} \in E$.
A sequence of vertices $P = (u_1, u_2, \ldots, u_k)$ is called a \emph{path} 
if $u_i$ is adjacent to $u_{i+1}$ for any $1 \le i < k$ and all the vertices are distinct.
A sequence of vertices $C = (u_1, u_2, \ldots, u_k)$ is called a \emph{cycle} if if $u_i$ is adjacent to $u_{i+1}$ for any $1 \le i \le k$, where $u_{k + 1}$ is considered as $u_1$, and all the vertices except for pair $\{u_1, u_{k+1}\}$ are distinct.
For $F \subseteq E$, we denote by $G[F]$ the subgraph consisting of all end vertices of $F$ and edges in $F$.
For two sets $X$ and $Y$, we denote by $X \triangle Y$ the symmetric difference  between $X$ and $Y$ (i.e., $X \triangle Y = (X \setminus Y) \cup (Y \setminus X)$).

Let $M$ be a set of edges in $G$.
We say that $M$ is a \emph{matching} of $G$ if for any pair of distinct $e, f \in M$, they does not share their end vertices (i.e., $e \cap f = \emptyset$ holds).
Moreover, $M$ is a \emph{maximal matching} of $G$ if $M$ is a matching and $M \cup  \set{e}$ is not a matching of $G$ for every $e \in E \setminus M$.
The maximum cardinality of a matching of $G$ is denoted by $\nu(G)$. 
Every matching with cardinality $\nu(G)$ is called a \emph{maximum matching} of $G$.
For a matching $M$, we say that a vertex $v$ is \emph{matched} in $M$ if $M$ has an edge incident to $v$.
Otherwise, $v$ is \emph{unmatched} in $M$.

In this paper, we measure the running time of enumeration algorithms in an \emph{output-sensitive manner}~\cite{Johnson:Yannakakis:Papadimitriou:IPL:1988}.
In particular, we focus on the delay of enumeration algorithms: The \emph{delay} of a enumeration algorithm is the maximum time interval between two consecutive outputs (including both preprocessing time and postprocessing time). 

\section{Hardness of the extension problem}\label{sec:hard}

We show that a direct application of the binary partition technique or the $k$-best enumeration framework to {\sc Large Maximal Matching Enumeration} seems to be impossible.

In enumeration algorithms based on the binary partition technique~\cite{Birmele:2013,Fukuda:1992}, we solve a certain decision or optimization problem, called an {\em extension problem}, to enumerate solutions.
Basically, the extension problem asks to decide whether, given disjoint sets $I$ and $O$, there is a solution that includes all elements in $I$ and excludes every element in $O$.
For enumerating \emph{maximum} matchings, the extension problem is tractable: For $I, O \subseteq E$ with $I \cap O = \emptyset$, the extension problem simply asks for a matching $M$ in a graph obtained by removing all endpoints in $I$ and edges in $O$ from $G$ with $|M| = \nu(G) - |I|$.
However, for {\sc Large Maximal Matching Enumeration}, the extension problem is intractable.
The formal definition of the extension problem is as follows:
Given a graph $G = (V, E)$ and $I, O \subseteq E$ with $I \cap O = \emptyset$, {\sc Maximal Matching Extension} asks to determine whether $G$ has a maximal matching $M$ with $I \subseteq M$ and $M \cap O = \emptyset$.

\begin{theorem}\label{theo:hard}
    \textsc{Maximal Matching Extension} is \NP-complete even on planar bipartite graphs with maximum degree three.
\end{theorem}

\begin{proof}
    Clearly \textsc{Maximal Matching Extension} is in \NP.
    To prove the \NP-hardness, we perform a polynomial-time reduction from a variant of {\sc Planar $3$-SAT}. 
    Let $\phi$ be a 3-CNF formula with variable set $X = \set{x_1, \ldots, x_n}$ and clause set $C = \set{C_1, \ldots, C_m}$.
    For a clause $C_j = (c_{j, 1} \lor c_{j, 2} \lor c_{j, 3})$, we call $c_{j, k}$ \emph{the $k$-th literal of $C_j$}.
    Note that $\phi$ may contain a clause with two literals.  
    The incidence bipartite graph $G_\phi$ of $\phi$ consists of vertex set $V_\phi = X \cup C$ and edge set $E_\phi = \inset{x_i, C_j}{x_i \in C_j}$. 
    It is known that {\sc 3-SAT} is \NP-complete even when $G_\phi$ is planar and each variable occurs twice positively and once negatively in $\phi$~\cite{Fellows:Kratochvil:Middendorf:Pfeiffer:Algorithmica:1995}. 
    
    \begin{figure}[t]
        \centering
        \begin{tikzpicture}[scale=0.75]
            \pgfdeclarelayer{nodelayer}
\pgfdeclarelayer{edgelayer}
\pgfsetlayers{nodelayer,edgelayer}
\tikzstyle{variable vertex}=[fill=none, draw=black, shape=circle, minimum size=2em, inner sep=0pt, outer sep=0pt]
\tikzstyle{clause vertex}  =[fill=none, draw=black, shape=circle, minimum size=2em, inner sep=0pt, outer sep=0pt]
\tikzstyle{forbidden edge}=[dashed]

\tikzstyle{assignment edge} = []

\begin{pgfonlayer}{nodelayer}
    \node [style=variable vertex] (x11) at (0, 1.5) {$x_1$};
    \node [style=variable vertex] (x12) at (1, 2.5) {$x'_1$};
    \node [style=variable vertex] (x13) at (2, 1.5) {$\neg x_1$};
    
    \node [style=variable vertex] (x21) at (4, 1.5) {$x_2$};
    \node [style=variable vertex] (x22) at (5, 2.5) {$x'_2$};
    \node [style=variable vertex] (x23) at (6, 1.5) {$\neg x_2$};
    
    \node [style=variable vertex] (x31) at (8,  1.5) {$x_3$};
    \node [style=variable vertex] (x32) at (9,  2.5) {$x'_3$};
    \node [style=variable vertex] (x33) at (10, 1.5) {$\neg x_3$};
    
    \node [style=variable vertex] (x41) at (12, 1.5) {$x_4$};
    \node [style=variable vertex] (x42) at (13, 2.5) {$x'_4$};
    \node [style=variable vertex] (x43) at (14, 1.5) {$\neg x_4$};
    
    \node [style=clause vertex] (c11) at (0,   -1) {$c_{1,1}$};
    \node [style=clause vertex] (c12) at (2,   -1) {$c_{1,2}$};
    \node [style=clause vertex] (c13) at (4,   -1) {$c_{1,3}$};
    \node [style=clause vertex] (d11) at (1, -2) {$d_1$};
    \node [style=clause vertex] (d12) at (3, -2) {$d'_1$};
    
    \node [style=clause vertex] (c21) at (6,   -1) {$c_{2,1}$};
    \node [style=clause vertex] (c22) at (8,   -1) {$c_{2,2}$};
    \node [style=clause vertex] (c23) at (10,   -1) {$c_{2,3}$};
    \node [style=clause vertex] (d21) at (7, -2) {$d_2$};
    \node [style=clause vertex] (d22) at (9, -2) {$d'_2$};
    
    \node [style=clause vertex] (c31) at (12,   -1) {$c_{3,1}$};
    \node [style=clause vertex] (c32) at (14,   -1) {$c_{3,2}$};
    \node [style=clause vertex] (d31) at (13, -2) {$d_3$};

\end{pgfonlayer}
\begin{pgfonlayer}{edgelayer}

		\draw (x11) to (x12);
		\draw (x12) to (x13);
		\draw (x21) to (x22);
		\draw (x22) to (x23);
		\draw (x31) to (x32);
		\draw (x32) to (x33);
		\draw (x41) to (x42);
		\draw (x42) to (x43);
		
		\draw (c11) to (d11); 
		\draw (d11) to (c12); 
		\draw (c12) to (d12); 
		\draw (d12) to (c13); 
		
		\draw (c21) to (d21); 
		\draw (d21) to (c22); 
		\draw (c22) to (d22); 
		\draw (d22) to (c23); 
		
		\draw (c31) to (d31); 
		\draw (d31) to (c32); 
		
		\draw[style=forbidden edge] (c11) to (x11); 
		\draw[style=forbidden edge] (c12) to (x23); 
		\draw[style=forbidden edge] (c13) to (x31); 
		\draw[style=forbidden edge] (c21) to (x13); 
		\draw[style=forbidden edge] (c22) to (x21); 
		\draw[style=forbidden edge] (c23) to (x41); 
		\draw[style=forbidden edge] (c31) to (x21); 
		\draw[style=forbidden edge] (c32) to (x43); 
		
        \draw[assignment edge] (x13) to (x12); 
        \draw[assignment edge] (x23) to (x22); 
        \draw[assignment edge] (x31) to (x32); 
        \draw[assignment edge] (x43) to (x42); 
        
		\draw[assignment edge] (c11) to (d11); 
		\draw[assignment edge] (c13) to (d12); 
		\draw[assignment edge] (d21) to (c22); 
		\draw[assignment edge] (d22) to (c23); 
		\draw[assignment edge] (c31) to (d31); 
\end{pgfonlayer}
        \end{tikzpicture}
        \caption{An example of $G'_\phi$ constructed from a 3-CNF formula $ \phi = (x_1\lor \lnot x_2 \lor x_3) \land (\lnot x_1 \lor x_2 \lor x_4) \land(x_2 \lor \lnot x_4)$.
        Edges in $O_\phi$ are depicted as dotted edges and $I_\phi$ is empty.
        Blue thick edges represent a maximal matching $M$ of $G'_\phi$ with $O_\phi \cap M = \emptyset$ corresponding to a truth assignment $(x_1, x_2, x_3, x_4) = (\texttt{False}, \texttt{False}, \texttt{True}, \texttt{False})$ of $\phi$. }
        \label{fig:hardness}
    \end{figure}
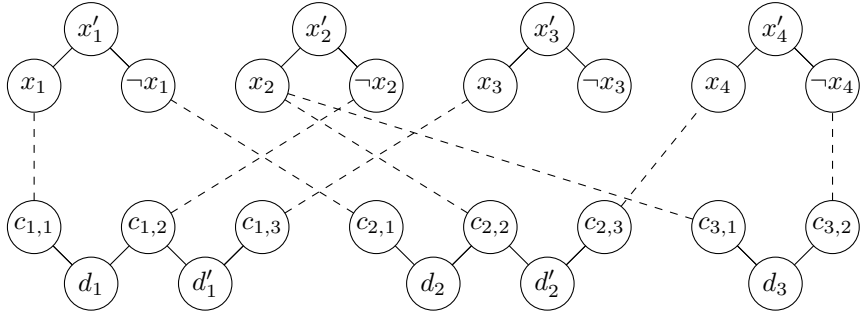
    
    Given an instance $\phi$ of the variant of {\sc 3-SAT}, we construct an instance $(G'_\phi, I_\phi, O_\phi)$ of \textsc{Maximal Matching Extension} so that 
    $\phi$ is satisfiable if and only if $(G'_\phi, I_\phi, O_\phi)$ is a yes-instance. 
    For each variable $x_i$, the variable gadget is a path $P_i = (x_i, x'_i, \neg x_i)$ of length two.
    For each clause $C_j$, 
    the clause gadget is a path $Q_j = (c_{j,1}, d_j, c_{j,2}, d'_j, c_{j,3})$ of length four if $C_j$ contains exactly three literals.
    Otherwise, the clause gadget is a path $Q_j = (c_{j,1}, d_j, c_{j,2})$ of length two.  
    $G'_\phi$ is obtained from these gadgets by adding an edge $\set{x_i, c_{j,k}}$ if $x_i$ is the $k$-th literal of $C_j$, and $\set{\neg x_i, c_{j,k}}$ if $\neg x_i$ is the $k$-th literal of $C_j$ for $1 \le i \le n$, $1 \le j \le m$, and $1 \le k \le 3$. See \Cref{fig:hardness} for an example. 
    Note that the maximum degree of $G'_\phi$ is three. 
    Since $G'_\phi$ is obtained from $G_\phi$ by replacing each vertex $x_i$ with $P_i$ and each vertex $C_i$ with $Q_i$, $G'_\phi$ is planar. 
    Moreover, since $G_\phi$ is bipartite, $G'_\phi$ is bipartite. 
    To complete the construction of the corresponding instance of \textsc{Maximal Matching Extension}, 
    let $I_\phi = \emptyset$ and let $O_\phi$ be the set of edges in $G'_\phi$ that lie between variable gadgets and clause gadgets.

    Suppose that $\phi$ has a satisfying truth assignment.
    From this assignment, we construct a matching $M$ of $G'_\phi$ as follows.
    For each variable $x_i$ that is assigned to true (resp. false),
    we add an edge $\set{x_i, x'_i}$ (resp. $\set{\neg x_i, x'_i}$) to $M$.
    For each clause $C_j$, at least one of its literals, say $c_{j, k}$, is true.
    Then, we add the unique maximum matching $M'$ in $Q_j$ to $M$ that has no edges incident to $c_{j, k}$.
    Note that any vertex $c_{j, k'}$ with $c_{j, k'} \neq c_{j, k}$ in $Q_j$ is matched in $M'$. 
    This completes the construction of $M$, and it is easy to verify that $M$ is a matching of $G'_\phi$ with $M \cap O_\phi = \emptyset$.
    Moreover, for any edge $e \in O_\phi$, at least one of endpoints of $e$ is matched in $M$. 
    Thus, if $\phi$ has a satisfying truth assignment, then $G'_\phi$ has a maximal matching that contains $M$ but none of any edge in $O_\phi$.
    
    Conversely, suppose that $M$ is a maximal matching of $G'_\phi$ with $M \cap O_\phi = \emptyset$.
    As $M \cap O_\phi = \emptyset$ and the maximality of $M$, $M$ contains exactly one edge in each variable gadget. 
    We construct an assignment of variables in such a way that for each variable $x_i$, $x_i$ is assigned to true if $\set{x_i, x'_i} \in M$ and to false otherwise. 
    Then, we show that this assignment is a satisfying truth assignment for $\phi$.
    To see this, consider a clause $C_j$.
    From the maximality of $M$, $Q_j$ contains exactly one vertex $c_{j, k}$ unmatched in $M$. 
    Suppose for contradiction that each literal of $C_j$ is false under the above assignment.
    Since $c_{j, k}$ is false, the corresponding variable vertex, say $x_i$, is unmatched in $M$.
    Thus, since both $c_{j, k}$ and $x_i$ are unmatched in $M$, $M \cup \set{\set{x_i, c_{j, k}}}$ or $M \cup \set{\set{\neg x_i, c_{j, k}}}$ is a matching of $G'_\phi$, contradicting to the maximality of $M$.
\end{proof}

As for $k$-best enumeration algorithms based on Lawler's framework~\cite{Lawler1972}, we need to solve an optimization version of {\sc Maximal Matching Extension}.
The above theorem also rules out the applicability of Lawler's framework to obtain a polynomial-delay algorithm for {\sc Large Maximal Matching Enumeration}, assuming that \P $\neq$ \NP.


\section{Enumeration of maximal matchings}\label{sec:main}

\subsection{{\sc Large Maximal Matching Enumeration}}
Our enumeration algorithm is based on the \emph{supergraph technique}, which is frequently used in many enumeration algorithms~\cite{Cohen::2008,Conte::2019,DBLP:conf/mfcs/KuritaK20,Schwikowski::2002,Khachiyan::2008,Khachiyan::2006:matroid}.
In particular, our algorithm is highly related to the enumeration algorithm for maximal independent sets with the {\em input-restricted problem} due to \cite{Cohen::2008}.
The basic idea of the supergraph technique is quite simple.
We define a directed graph $\mathcal G$ whose node set corresponds to all the solutions we wish to enumerate.
The enumeration algorithm solely traverses this directed graph and outputs a solution at each node.
To this end, we need to carefully design the arc set of $\mathcal G$ so that all the nodes in $\mathcal G$ are reachable from a specific node.

For maximal matchings (without cardinality constraints), we can enumerate those in polynomial delay with this technique.
Let $G = (V, E)$ be a graph.
For a (not necessarily maximal) matching $M$ of $G$, we denote by $\comp{M}$ an arbitrary maximal matching of $G$ that contains $M$.
This maximal matching can be computed from $M$ by greedily adding edges in $E \setminus M$.
Let $M$ be a maximal matching of $G$.
For $e \in E \setminus M$, $(M \setminus \Gamma(e)) \cup \set{e}$ is a matching of $G$, and then $M_e = \comp{(M \setminus \Gamma(e)) \cup \set{e}}$ is a maximal matching of $G$.
We define the (out-)neighbors of a maximal matching $M$ in $\mathcal G$, denoted $\mathcal N_{\mathcal G}(M)$, as the set of maximal matchings $\set{M_e : e \in E \setminus M}$.
To avoid a confusion, each maximal matching in $\mathcal N_{\mathcal G}(M)$ is called a {\em $\mathcal G$-neighbor} of $M$.
The arc set of $\mathcal G$ is defined by this neighborhood relation.
With this definition, we can show that every maximal matching $M_2$ is reachable from any other maximal matching $M_1$, that is, $\mathcal G$ is strongly connected.
To see this, we consider the value $m(M, M') = |M \cap M'|$ defined between two (maximal) matchings $M$ and $M'$ of $G$.
Since $M_1$ and $M_2$ are maximal matchings of $G$, there is an edge $e \in M_2 \setminus M_1$.
Then, we have
\begin{linenomath}
\[
    m(M_1, M_2) = \size{M_1 \cap M_2} < \size{((M_1 \setminus \Gamma(e)) \cup \set{e}) \cap M_2}  \le \size{\comp{(M_1 \setminus \Gamma(e)) \cup \set{e}} \cap M_2},
\]
\end{linenomath}
where the first inequality follows from $e \in M_2$ and $\Gamma(e) \cap M_2 = \emptyset$ and the second inequality follows from $M \subseteq \comp{M}$ for any matching $M$ of $G$.
This indicates that $M_1$ has a $\mathcal G$-neighbor $M' = \comp{(M_1 \setminus \Gamma(e)) \cup \{e\}}$ such that $m(M_1, M_2) < m(M', M_2)$.
Moreover, the following proposition holds.

\begin{proposition}\label{prop:identical}
    Let $M_1$ and $M_2$ be maximal matchings of $G$.
    Then, $m(M_1, M_2) \le \size{M_2}$.
    Moreover, $m(M_1, M_2) = \size{M_2}$ if and only if $M_1 = M_2$.
\end{proposition}

By induction on $k = \size{M_2} - m(M_1, M_2)$, there is a directed path from $M_1$ to $M_2$ in $\mathcal G$ for every pair of maximal matchings of $G$, which proves the strong connectivity of $\mathcal G$.

Our algorithm for {\sc Large Maximal Matching Enumeration} also traverses $\mathcal G$ from an arbitrary {\em maximum} matching of $G$ in a breadth-first manner but truncates all maximal matchings of cardinality less than given threshold $t$.
The pseudocode is shown in \Cref{algo:sg}.
In the following, we show that \Cref{algo:sg} enumerates all the maximal matchings of $G$ with cardinality at least $t$, provided $t < \nu(G)$.
We discuss later for the other case $t = \nu(G)$.

\DontPrintSemicolon
\begin{algorithm}[t]
    \caption{Given a graph $G$ and an integer $t$, the algorithm enumerates all maximal matchings of $G$ with cardinality at least $t$.}
    \label{algo:sg}
    \Procedure{\Traverse{$G, t$}}{
        Let $M^*$ be a maximum matching of $G$\;
        Add $M^*$ to a queue $\mathcal Q$ and to set $\mathcal S$\;
        \While{$\mathcal Q$ is not empty}{
            Let $M$ be a maximal matching in $\mathcal Q$\;
            Output $M$ and delete $M$ from $\mathcal Q$\;
            \ForEach{$M' \in \mathcal N_{\mathcal G}(M)$}{
                \lIf{$M' \not\in \mathcal S$ and $\size{M'} \ge t$}{
                    Add $M'$ to $\mathcal Q$ and to $\mathcal S$
                }
            }
        }
    }
\end{algorithm}

For a non-negative integer $k$, we say that a directed path in $\mathcal G$ is {\em $k$-thick} if every maximal matching on the path has cardinality at least $k$.
To show the correctness of \Cref{algo:sg}, it is sufficient to prove that 
(1) for any pair of maximum matching $M^*$ of $G$ and a maximal matching $M$ of $G$ with $\size{M} < \nu(G)$, there is a directed $\size{M}$-thick path from $M^*$ to $M$ in $\mathcal G$ and 
(2) for a pair of maximum matchings $M$ and $M'$ of $G$, there is a directed $(\nu(G) - 1)$-thick path from $M$ to $M'$ in $\mathcal G$. 

In the rest of this subsection, fix distinct maximal matchings $M_1$ and $M_2$ of $G$.
Since $M_1$ and $M_2$ are matchings of $G$, each component of the graph $G[M_1 \triangle M_2]$ is either a path or a cycle.
We say that a path component $P$ in $G[M_1 \triangle M_2]$ is \emph{even-alternating} if exactly one end vertices of $P$ is unmatched in $M_1$.
Let us note that $P$ is even-alternating if and only if it has an even number of edges.
We say $P$ is \emph{$M_1$-augmenting} (resp. $M_2$-augmenting) if the both end vertices of $P$ are unmatched in $M_1$ (resp. $M_2$).
Since both $M_1$ and $M_2$ are maximal matchings of $G$, the following proposition holds.

\begin{proposition}\label{prop:has-two-edges}
    Every component in $G[M_1 \triangle M_2]$ is either a path with at least two edges or a cycle with at least four edges.
\end{proposition}

For $M_2$-augmenting path component $P$ in $G[M_1 \triangle M_2]$, it holds that $\size{M_1 \cap E(P)} > \size{M_2 \cap E(P)}$.
Thus, the following proposition holds.
\begin{proposition}
    \label{prop:M1:is:larger:than:M2}
    Suppose that $G[M_1 \triangle M_2]$ has no $M_1$-augmenting path components but has at least one $M_2$-augmenting path component. 
    Then, $\size{M_1} > \size{M_2}$. 
\end{proposition}

\begin{lemma}\label{lem:goodpath}
    If $G[M_1 \triangle M_2]$ has an $M_1$-augmenting or even-alternating path component $P$, then there is a directed $\size{M_1}$-thick path from $M_1$ to a maximal matching $M'$ in $\mathcal G$ such that $m(M_1, M_2) < m(M', M_2)$ and $\size{M'} \ge \size{M_1}$.
    Moreover, if $P$ is $M_1$-augmenting, then $\size{M'} > \size{M_1}$.
\end{lemma}
\begin{proof}
    Let $P = (v_1, v_2, \ldots, v_\ell)$ and let $e_i = \{v_i, v_{i+1}\}$ for $1 \le i < \ell$. 
    By \Cref{prop:has-two-edges}, $P$ contains at least two edges.
    Assume, without loss of generality, we have $e_1 \in M_2$ and $e_2 \in M_1$. 
    Define $\hat{M} = (M_1 \setminus \set{e_2}) \cup \set{e_1}$.
    As $v_1$ is unmatched in $M_1$. $\hat{M}$ is a matching of $G$, and hence $M' = \comp{\hat{M}}$ is a $\mathcal G$-neighbor of $M_1$. 
    Then,
    \begin{linenomath}
    \[
        m(M', M_2)  \ge m(\hat{M}, M_2) = \size{((M_1 \setminus \set{e_2}) \cup \set{e_1}) \cap M_2} > \size{M_1 \cap M_2} = m(M_1, M_2),
    \]
    \end{linenomath}
    as $e_1 \in M_2$ and $e_2 \in M_1 \setminus M_2$.
    Moreover, we have $\size{\comp{\hat{M}}} \ge \size{\hat{M}} = \size{M_1}$. 
    Thus, the arc $(M_1, \comp{\hat{M}})$ is the desired directed path in $\mathcal G$. 
    
    Suppose moreover that $P$ is $M_1$-augmenting.
    In this case, we have $\ell \ge 4$.
    We prove the claim by induction on $\ell$.
    Let $\hat{M}$ be as above.
    If $\hat{M} \subset \comp{\hat{M}}$, we are done.
    Suppose otherwise.
    If $\ell = 4$, as $v_4$ is unmatched in $M_1$, $\hat{M} \cup \{e_3\}$ is a matching of $G$, implying that $\hat{M} \subset \comp{\hat{M}}$.
    Otherwise, that is, $\hat{M} = \comp{\hat{M}}$, the subpath $(v_3, v_4, \dots, v_\ell)$ of $P$ is a path component in $G[\hat{M} \triangle M_2]$ and is $\hat{M}$-augmenting.  
    Applying the induction hypothesis to this subpath, there is a directed $\size{M'}$-thick path from $M'$ to a maximal matching $M''$ in $\mathcal G$ such that $m(M', M_2) < m(M'', M_2)$ and $\size{M''} > \size{M'} = \size{M_1}$.
    As $m(M_1, M_2) < m(M', M_2) < m(M'', M_2)$, the lemma follows.
\end{proof}

If $G[M_1 \triangle M_2]$ has neither $M_1$-augmenting path components nor cycle components, the maximal matching $M'$ of $G$ in \Cref{lem:goodpath} satisfies the following additional property.

\begin{corollary}\label{cor:goodpath}
    Let $M'$ be the maximal matching of $G$ obtained in \Cref{lem:goodpath}.
    If $G[M_1 \triangle M_2]$ has neither $M_1$-augmenting path components nor cycle components, then $G[M' \triangle M_2]$ has no cycle components.
\end{corollary}
\begin{proof}
    Let $P$ be an even-alternating path component in $G[M_1 \triangle M_2]$.
    Then, $M'$ is defined to be $\comp{\hat{M}}$, where $\hat{M} = (M_1 \setminus \{e_2\}) \cup \{e_1\}$.
    As $\hat{M} \triangle M_2 = (M_1 \triangle M_2) \setminus \set{e_1, e_2}$, $G[\hat{M} \triangle M_2]$ has no $\hat{M}$-augmenting path components.
    If $G[M' \triangle M_2]$ has a cycle component $C$, then there must be an $\hat{M}$-augmenting path component $P'$ in $G[\hat{M} \triangle M_2]$ such that $P'$ together with some $e \in M' \setminus \hat{M}$ forms the cycle $C$, which is a contradiction.
\end{proof}

\begin{lemma}\label{lem:cycle}
    If $G[M_1 \triangle M_2]$ has a cycle component $C$, then there is a directed $(\size{M_1} - 1)$-thick path from $M_1$ to a maximal matching $M'$ of $G$ in $\mathcal G$ such that $m(M_1, M_2) < m(M', M_2)$ and $\size{M'} \ge \size{M_1}$. 
\end{lemma}
\begin{proof}
    Let $C = (v_1, v_2, \ldots, v_\ell)$ and for each $1 \le i \le \ell$, let $e_i = \{v_i, v_{i+1}\}$, where $v_{\ell + 1} = v_1$.
    Assume without loss of generality that $e_i \in M_1$ for odd $i$.
    Define $\hat{M} = (M_1 \setminus \set{e_1, e_3}) \cup \set{e_2}$.
    Clearly, $\hat{M}$ is a matching of $G$ and then $M' = \comp{\hat{M}}$ is a $\mathcal G$-neighbor of $M_1$.
    Similarly to \Cref{lem:goodpath}, we have $m(M_1, M_2) < m(M', M_2)$.
    If $\size{M_1} \le \size{M'}$, we are done.
    Moreover, if $\ell = 4$, $\hat{M} \cup \set{e_4}$ is a matching of $G$ and hence we have $\size{M_1} \le \size{M'}$ as well.
    Thus, suppose that $\ell \ge 6$ and $\size{M'} = \size{\hat{M}} = \size{M_1} - 1$.
    The subpath $P = (v_4, v_5, \ldots, v_{\ell+1})$ is an $M'$-augmenting path component in $G[M' \triangle M_2]$.
    By \Cref{lem:goodpath}, there is a directed path from $M'$ to a maximal matching $M''$ of $G$ in $\mathcal G$ such that $m(M', M_2) < m(M'', M_2)$ and $\size{M'} < \size{M''}$.
    Moreover, each maximal matching on the directed path has cardinality at least $\size{M'}$.
    This completes the proof of lemma.
\end{proof}

\begin{lemma}
    \label{lem:badpath}
     If $G[M_1 \triangle M_2]$ has an $M_2$-augmenting path component $P$, then there is a directed $(\size{M_1} - 1)$-thick path from $M_1$ to a maximal matching $M'$ of $G$ in $\mathcal G$ such that $m(M_1, M_2) < m(M', M_2)$ and $\size{M'} \ge \size{M_1} - 1$. 
\end{lemma}
\begin{proof}
    The proof is almost analogous to that in \Cref{lem:goodpath}.
    Let $P = (v_1, v_2, \ldots, v_\ell)$ and let $e_i = \{v_i, v_{i+1}\}$ for $1 \le i < \ell$. 
    By \Cref{prop:has-two-edges}, $P$ contains at least two edges.
    Moreover, as $P$ is $M_2$-augmenting, $P$ contains at least three edges and $e_1, e_3 \in M_1$ and $e_2 \in M_2$. 
    Define $\hat{M} = (M_1 \setminus \set{e_1, e_3}) \cup \set{e_2}$.
    Then, $\hat{M}$ is a matching of $G$, and hence $M' = \comp{\hat{M}}$ is a $\mathcal G$-neighbor of $M_1$.
    As $e_1, e_3 \in M_1 \setminus M_2$ and $e_2 \in M_2 \setminus M_1$, we have $m(M_1, M_2) < m(M', M_2)$.
    Moreover, $\size{\hat{M}} = \size{M_1} - 1$. the lemma follows.
\end{proof}

\begin{corollary}\label{cor:badpath}
    Let $M'$ be the maximal matching of $G$ obtained in \Cref{lem:badpath}.
    If $G[M_1 \triangle M_2]$ has neither $M_1$-augmenting path components nor cycle components, then $G[M' \triangle M_2]$ has no cycle components.
\end{corollary}
\begin{proof}
    The proof is almost identical to that in \Cref{cor:goodpath}.
    Since $G[\hat{M} \triangle M_2]$ has no $\hat{M}$-augmenting path components, $G[M' \triangle M_2]$ has no cycle components.
\end{proof}

Now, we are ready to prove the main claims.

\begin{lemma}\label{lem:large-to-small}
    Let $M_1$ and $M_2$ be maximal matchings of $G$ with $\size{M_1} > \size{M_2}$.
    Then, there is a directed $\size{M_2}$-thick path from $M_1$ to $M_2$ in $\mathcal G$.
\end{lemma}
\begin{proof}
    We prove the lemma by induction on $k = \size{M_2} - m(M_1, M_2)$.
    Note that, by \Cref{prop:identical}, it holds that $k \ge 0$.
    Let $G' = G[M_1 \triangle M_2]$.
    We prove a slightly stronger claim: If either (1) $\size{M_1} > \size{M_2}$ or (2) $\size{M_1} = \size{M_2}$ and $G'$ has no cycle components,
    then there is a directed $\size{M_2}$-thick path from $M_1$ to $M_2$ in $\mathcal G$.
    The base case $k = 0$ follows from \Cref{prop:identical}.
    
    We assume that $k > 0$ and the lemma holds for all $k' < k$.
    As $k > 0$, $G'$ has at least one connected component.
    Suppose that $G'$ has a cycle component.
    In this case, it holds that $\size{M_1} > \size{M_2}$ from the assumption.
    By \Cref{lem:cycle}, $\mathcal G$ has a directed $(\size{M_1} - 1)$-thick path from $M_1$ to a maximal matching $M'$ of $G$ such that $m(M_1, M_2) < m(M', M_2)$ and $\size{M'} \ge \size{M_1}$.
    Applying the induction hypothesis to $M'$ and $M_2$, $\mathcal G$ has a directed $\size{M_2}$-thick path from $M'$ to $M_2$.
    As $\size{M_1} > \size{M_2}$, the directed path obtained by concatenating these two paths (from $M_1$ to $M'$ and from $M'$ to $M_2$) is a directed $\size{M_2}$-thick path from $M_1$ to $M_2$ and hence the claim follows in this case.
    
    Suppose that $G'$ has an $M_1$-augmenting path component.
    By \Cref{lem:goodpath}, $\mathcal G$ has a directed $\size{M_1}$-thick path from $M_1$ to $M'$ such that
    $m(M_1, M_2) < m(M', M_2)$ and $\size{M_1} < \size{M'}$.
    Since pair $M'$ and $M_2$ satisfies (1), by the induction hypothesis, $\mathcal G$ has a directed $\size{M_2}$-thick path from $M'$ to $M_2$, and hence the claim holds for this case.
    
    In the following, we assume that $G'$ has neither cycle components nor $M_1$-augmenting path components.
    Suppose that $G'$ has an even-alternating path component.
    By \Cref{lem:goodpath}, $\mathcal G$ has a directed $\size{M_1}$-thick path from $M_1$ to $M'$ such that
    $m(M_1, M_2) < m(M', M_2)$ and $\size{M_1} \le \size{M'}$.
    Moreover, by \Cref{cor:goodpath}, $G[M' \triangle M_2]$ has no cycle components.
    Thus, applying the induction hypothesis to $M'$ and $M_2$ proves the claim.
    
    Finally, suppose that $G'$ has only $M_2$-augmenting path components.
    By \Cref{prop:M1:is:larger:than:M2}, $\size{M_1} > \size{M_2}$.
    By \Cref{lem:badpath}, $\mathcal G$ has a directed $(\size{M_1} - 1)$-thick path from $M_1$ to $M'$ such that $m(M_1, M_2) < m(M', M_2)$ and $\size{M'} \ge \size{M_1} - 1$.
    Moreover, by \Cref{cor:badpath}, $G[M' \triangle M_2]$ has no cycle components.
    Thus, as $\size{M'} \ge \size{M_2}$, applying the induction hypothesis to $M'$ and $M_2$ proves the claim as well.
\end{proof}

\begin{lemma}\label{lem:maximum-to-maximum}
    Let $M_1$ and $M_2$ be maximum matchings of $G$.
    Then, there is a directed $(\nu(G) - 1)$-thick path from $M_1$ to $M_2$ in $\mathcal G$.
\end{lemma}
\begin{proof}
    We prove the lemma by induction on $k = \size{M_2} - m(M_1, M_2)$.
    The base case $k = 0$ follows from \Cref{prop:identical}.
    We assume that $k > 0$ and the lemma holds for all $k' < k$.
    Let $G' = G[M_1 \triangle M_2]$.
    As $k > 0$, $G'$ has at least one connected component.
    Observe that every component of $G'$ is either an even-alternating component or a cycle component.
    This follows from the fact that if $G'$ has an $M_1$- or $M_2$-augmenting path component, then this path is an augmenting path for $M_1$ or $M_2$, respectively, which contradicts to the assumption that $M_1$ and $M_2$ are maximum matchings of $G$.
    By \Cref{lem:goodpath,lem:cycle}, $\mathcal G$ has a directed $(\size{M_1} - 1)$-thick path from $M_1$ to a maximal matching $M'$ of $G$ such that $m(M_1, M_2) < m(M', M_2)$ and $\size{M_1} \le \size{M'}$.
    As $M_1$ is a maximum matching of $G$, $M'$ is also a maximum matching of $G$.
    Applying the induction hypothesis to pair $M'$ and $M_2$ proves the lemma.
\end{proof}

Thus, we can enumerate all large maximal matchings in polynomial delay.
By simply traversing $\mathcal G$, we can enumerate all neighbor of $\mathcal G$ in $\order{nm}$ time.
However, to determine whether each neighbor has already been output, we need a data structure.

\begin{lemma}\label{lem:dup}
    Let $\mathcal M$ be a collection of maximal matchings.
    There is a data structure for representing $\mathcal M$ that supports the following operations:
    (1) Decide if $\mathcal M$ contains a given matching $M$ in $\order{n}$ time;
    (2) Insert a matching $M$ into $\mathcal M$ in $\order{n}$ time.
\end{lemma}
\begin{proof}
    In this data structure, the set of vertices in $G$ is regarded as $n$ distinct integers $\set{1, 2, \ldots, n}$.
    Then, each edge $\{u, v\} \in E$ can be seen as an ordered pair of integers $(u, v)$ with $u < v$.
    We also define a natural lexicographical order on the edges of $E$: For two edges $e_1 = (u_1, v_1)$ and $e_2 = (u_2, v_2)$, $e_1$ is smaller than $e_2$ if and only if $u_1 < u_2$ or $u_1 = u_2$ and $v_1 < v_2$.
    Under this edge ordering, we can sort the edges in a matching $M = \{(u_1, v_1), (u_2, v_2), \ldots, (u_k, v_k)\}$ in $\order{n}$ time with radix sort.
    From the sorted edge sequence $((u_1, v_1), (u_2, v_2), \ldots, (u_k, v_k))$, we can uniquely encode a matching into the sequence of integers $(u_1, v_1, u_2, v_2, \ldots, u_k, v_k)$ and then regard each matching $M$ as such a sequence of integers.
    With a standard radix trie data structure~\cite{Morrison:PATRICIA:1968}, we can check if a given sequence in the data structure and insert a new sequence into the data structure in time $\order{\ell + |\Sigma|}$, where $\ell$ is the maximum length of a sequence and $\Sigma$ is an alphabet used in sequences, which is $O(n)$ for our sequences.
\end{proof}

\begin{theorem}\label{thm:nu>t}
    \Cref{algo:sg} enumerates all maximal matchings of $G$ with cardinality at least $t$ in $\order{nm}$ delay and exponential space, provided that $t < \nu(G)$.
\end{theorem}
\begin{proof}
    The correctness of the algorithm directly follows from \Cref{lem:large-to-small,lem:maximum-to-maximum}.
    We analyze the delay of the algorithm. 
    We first compute a maximum matching $M^*$ of $G$.
    This can be done in time $O(n^{1/2}m)$ using the algorithm of \cite{DBLP:conf/focs/MicaliV80}.
    Each maximal matching has $\mathcal G$-neighbors at most $m$.
    For each such $\mathcal G$-neighbor $M'$ of $M$, we can check whether $M' \in \mathcal S$ at line 8 in $O(n)$ time with the data structure given in \Cref{lem:dup}.
    Thus, it suffices to show that $\comp{(M \setminus \Gamma(e)) \cup \{e\}}$ can be computed in $O(n)$ time from given a maximal matching $M$ and $e \in E \setminus M$.
    Observe that $M \cap \Gamma(e)$ consists of at most two edges $f_1, f_2 \in E$, each of which is incident to one of the end vertices of $e$.
    Since at least one of end vertices of each edge in $E \setminus (\Gamma(f_1) \cup \Gamma(f_2))$ is matched in $(M \setminus \set{f_1, f_2}) \cup \{e\}$, we can compute $\comp{(M \setminus \Gamma(e)) \cup \{e\}}$ from $(M \setminus \Gamma(e)) \cup \{e\}$ by greedily adding edges in $\Gamma(f_1) \cup \Gamma(f_2)$.
    Since $\size{\Gamma(f_1) \cup \Gamma(f_2)} = O(n)$, this can be done in $O(n)$ time.
    Therefore, the theorem follows.
\end{proof}

To reduce the space complexity, we follow another well-known strategy, called the \emph{reverse search technique}, due to Avis and Fukuda~\cite{Avis::1996}.
The basic idea of this technique is to define a rooted tree $\mathcal T$ over the set of solutions instead of a directed graph.
The reverse search technique solely traverses this rooted tree and outputs a solution on each node in the tree.
A crucial difference from the supergraph technique is that we do not need exponential-space data structures used in the supergraph technique to avoid duplicate outputs.


\newcommand{\reconf}{\mathcal A}
\newcommand{\parent}{{\texttt{par}}}

Let $\mathcal G_{\ge t}$ be the subgraph of $\mathcal G$ induced by the maximal matchings of $G$ with cardinality at least $t$.
Our rooted tree $\mathcal T$ is in fact defined as a spanning tree of the underlying undirected graph of $\mathcal G_{\ge t}$.
To define the rooted tree $\mathcal T$, we select an arbitrary maximum matching $R^*$ of $G$ as a root.
In the following, we fix $R^*$ and define a parent function $\parent$ with respect to $R^*$.
We assume that the edges in $G$ are totally ordered with respect to some edge ordering.
Let $M$ be a maximal matching of $G$ with $\size{M} \ge t$ and $M \neq R^*$.
If $G[M\triangle R^*]$ contains a path component, then there is an edge $e \in R^* \setminus M$ that is incident to an end vertex of a path component in $G[M\triangle R^*]$.
Then, we choose the minimum edge (with respect to the edge ordering) satisfying this condition as $e$, and we define $\parent(M) = \comp{(M \setminus \Gamma(e)) \cup \{e\}}$.
Otherwise, we choose  the minimum edge in $R^* \setminus M$ as $e$, and we define $\parent(M) = \comp{(M \setminus \Gamma(e)) \cup \{e\}}$.
Note that as $R^*$ is a maximum matching of $G$, there is at least one path component in $G[M \triangle R^*]$ whose end vertex is matched in $R^*$ under the assumption that $G[M \triangle R^*]$ has a path component. 
Similarly to the proofs of \Cref{lem:goodpath,lem:cycle}, we have $m(\parent(M), R^*) > m(M, R^*)$ and $|\parent(M)| \ge \min\{|M|, \nu(G) - 1\}$.
This implies that the parent function defines a rooted tree $\mathcal T$ in $\mathcal G_{\ge t}$ with root $R^*$ by considering $\parent(M)$ is the parent of $M$.

Next, we consider the time complexity for computing $\parent(M)$.
As we have seen in the proof of \Cref{thm:nu>t}, $\comp{(M\setminus \Gamma(e)) \cup \set{e}}$ can be computed in $\order{n}$ time.
Given $M$ and $R^*$, we can compute the minimum edge $e$ in $G[M \triangle R^*]$ in $\order{n}$ time as $|M| + |R^*| \le n$.
Thus, $\parent(M)$ can be computed in $\order{n}$ time and the following lemma holds.

\begin{lemma}\label{lem:suc}
    Given a maximal matching $M$ of $G$ with $M \neq R^*$, we can compute $\emph{\parent}(M)$ in $\order{n}$ time.
\end{lemma}

Now, we are ready to describe our polynomial-space enumeration algorithm for {\sc Large Maximal Matching Enumeration}.
For maximal matchings $M$ and $M'$ of $G$ with cardinality at least $t$, $M'$ is a \emph{child} of $M$ if $\parent(M') = M$.
\Cref{algo:rs} recursively generates the set of children of a given maximal matching, which enable us to traverse all nodes in $\mathcal T$.

\begin{algorithm}[t]
    \caption{Given a graph $G$ and an integer $t < \nu(G)$, {\tt Reverse-Search} enumerates all maximal matchings of $G$ with cardinality at least $t$.}
    \label{algo:rs}
    \Procedure{\tt Reverse-Search($G, t$)}{
        Let $R^*$ be a maximum matching of $G$\;
        \Traversetree{$G, t, R^*$}\;
    }
    \Procedure{\Traversetree{$G, t, M$}}{
        Output $M$\;
        Let $\mathcal C$ be the children of $M$\label{algo:rs:enum:children}\;
        \ForEach{$M' \in \mathcal C$}{
            \Traversetree{$G, t, M'$}
        }
    }
\end{algorithm}

The following lemma is vital to bound the delay of the algorithm.

\begin{lemma}
    \label{lem:comp:children}
    Let $M$ be a maximal matching of $G$ with $\size{M} \ge t$. 
    Then, there are $O(n\Delta^2)$ children of $M$.
    Moreover, the children of $M$ can be enumerated in total time $\order{n^2\Delta^2}$.
\end{lemma}
\begin{proof}
    Let $M'$ be a child of $M$.
    From the definition of the parent-child relation, we have
    \begin{align*}
        \size{M \triangle M'} = \size{\comp{(M' \setminus \Gamma(e)) \cup \{e\}} \triangle M'} \le 5.
    \end{align*}
    This inequality follows from the facts that $\size{M' \cap \Gamma(e)} \le 2$ and $\size{M \setminus \comp{(M' \setminus \Gamma(e)) \cup \{e\}}} \le 2$ (as observed in the proof of \Cref{thm:nu>t}).
    Thus, we can enumerate all the children $\mathcal C = \inset{M'}{\parent(M') = M}$ of $M$ in polynomial time.
    To improve the running time of enumerating children in $\mathcal C$, we take a closer look at both $M$ and $M'$.
    
    Let $e$ be the edge such that $M = \parent(M') = \comp{(M' \setminus \Gamma(e)) \cup \set{e}}$ and let $F = M'\cap\Gamma(e)$. 
    As observed above, $\size{F} \le 2$.
    Let $A$ be the edge set such that $M = (M' \setminus F) \cup \set{e} \cup A$.
    Each $f \in A$ is incident to an edge $f'$ in $F$ as $M'$ is a maximal matching of $G$.
    This implies that $A \subseteq M$ is uniquely determined from $e \in M$ and $F \subseteq \Gamma(e)$ with $\size{F} \le 2$ (i.e., $A = \inset{f \in (M \setminus \{e\}) \cap \Gamma(f')}{f' \in F}$).
    Thus, we have $|\mathcal C| = O(n\Delta^2)$.
    For each $e \in M$ and $F \subseteq \Gamma(e)$ with $|F| \le 2$, we can compute $M'$ as $M' = M \setminus (A \cup \{e\}) \cup F$ in $\order{n}$ time and, by \Cref{lem:suc}, check if $\parent(M') = M$ in $\order{n}$ time as well.
    Hence we can enumerate all children in $\mathcal C$ in $\order{n^2\Delta^2}$ time.
    %
\end{proof}

We output a solution for each recursive call.
This implies that the delay of the algorithm is upper bounded by the running time of \Cref{algo:rs:enum:children} in~\Cref{algo:rs}.
Since $m \le n\Delta$, the delay of the algorithm is bounded by $\order{n^{1/2}m + n^2\Delta^2} = \order{n^2\Delta^2}$ as well.
Moreover, as the depth of $\mathcal T$ is at most $|R^*|$, the algorithm runs in polynomial space.

Finally, when $t = \nu(G)$, we can enumerate all maximum matchings using the binary partition technique in $\order{nm}$ delay and polynomial space.
By combining these results, we obtain the following theorem.


\begin{theorem}
     One can enumerate all maximal matchings of $G$ with cardinality at least $t$ in $\order{n^2\Delta^2}$ delay with polynomial space.
\end{theorem}

\subsection{{\sc \texorpdfstring{$k$}{k}-Best Maximal Matching Enumeration}}
In the previous subsection, we define the graph $\mathcal G$, allowing us to enumerate all the maximal matchings of $G$ with cardinality at least given threshold $t$. 
The key to this result is \Cref{lem:large-to-small}, which states that, for any $t < \nu(G)$, there is a directed $t$-thick path from a maximum matching $M_1$ to a maximal matching $M_2$ with cardinality $t$.
This implies that every maximal matching of cardinality at least $t$ is ``reachable'' from a maximum matching in the subgraph of $\mathcal G$ induced by the node set $\{M: M \text{ is a maximal matching of } G, |M| \ge t\}$ for every $t < \nu(G)$.
From this fact, we can extend \Cref{algo:sg} to an algorithm for {\sc $k$-Best Maximal Matching Enumeration}, which is shown in \Cref{algo:k-best}.
The algorithm first enumerates all maximum matchings of $G$ with the algorithm $\mathcal A$ and outputs those maximum matchings as long as at most $k$ solutions are output.
The remaining part of the algorithm is almost analogous to \Cref{algo:sg} and the essential difference from it is that the algorithm chooses a largest maximal matching in the priority queue $\mathcal Q$ at line 9.
Intuitively, we traverse a forest based on best-first manner.

\begin{algorithm}[t]
    \caption{Given a graph $G$ and a non-negative integer $k$, the algorithm solves {\sc $k$-Best Maximal Matching Enumeration}}
    \label{algo:k-best}
    \Procedure{\tt $k$-best($G, k$)}{
        Let $\mathcal A$ be a polynomial delay enumeration algorithm for maximum matchings.\;
        \ForEach{$M$ generated by $\mathcal A(G)$}{
            Output $M$ and add $M$ to queue $\mathcal Q$\;
            \lIf{$k$ solutions are output}{{\bf halt}}
            \ForEach{$M' \in \mathcal N_{\mathcal G}(M)$ with $M' \notin \mathcal S$}{
                Add $M'$ to $\mathcal Q$ and to $\mathcal S$
            }
        }
        \While{$\mathcal Q$ is not empty}{
            Let $M$ be a largest maximal matching in $\mathcal Q$\;
            Output $M$ and delete $M$ from $\mathcal Q$\;
            \lIf{$k$ solutions are output}{{\bf halt}}
            \ForEach{$M' \in \mathcal N_{\mathcal G}(M)$ with $M' \notin \mathcal S$}{
                Add $M'$ to $\mathcal Q$ and to $\mathcal S$
            }
        }
    }
\end{algorithm}

\begin{theorem}\label{theo:kb:time}
    We can solve {\sc $k$-Best Maximal Matching Enumeration} in $O(nm)$ delay.
\end{theorem}
\begin{proof}
    The delay of the algorithm follows from a similar analysis in \Cref{thm:nu>t}.
    Note that, at line 9, we can choose in time $O(1)$ a largest maximal matching in $\mathcal Q$ by using $\nu(G)$ linked lists $L_1, L_2, \ldots, L_{\nu(G)}$ for $\mathcal Q$, where $L_i$ is used for maximal matchings of cardinality $i$.
    
    Let $\mathcal M$ be the set of maximal matchings of $G$ that are output by the algorithm.
    To show the correctness of the algorithm, suppose for contradiction that there are maximal matchings $M \notin \mathcal M$ and $M' \in \mathcal M$ of $G$ such that $\size{M} > \size{M'}$.
    Since $\mathcal G$ has a directed $\size{M}$-thick path from a maximum matching of $G$ to $M$, we can choose such $M$ so that every maximal matching on the path except for $M$ belongs to $\mathcal M$.
    Let $M''$ be the immediate predecessor of $M$ on the path.
    As $M'' \in \mathcal M$, $M$ must be in the queue $\mathcal Q$ at some point.
    Hence as $\size{M} > \size{M'}$, the algorithm outputs $M$ before $M'$.
\end{proof}

\section{Enumerating maximum matchings} 

To handle the other case $t = \nu(G)$, as mentioned in \Cref{sec:hard}, we solve {\sc Large Maximal Matching Enumeration} with the binary partition technique.
In this technique, we recursively partition the set $\mathcal M(G)$ of all maximum matchings of $G$ into subsets. 
The recursive procedure is defined as follows:
For a pair $I$ and $O$ of disjoint subsets of edge set $E$,
let $\mathcal M(G, I, O)$ be the subset of $\mathcal M(G)$ consisting of all maximum matchings $M$ of $G$ that satisfies $I \subseteq M$ and $M \cap O = \emptyset$. 
Initially, we set $I = O = \emptyset$ and hence we have $\mathcal M(G, I, O) = \mathcal M(G)$. 
Given $I, O \subseteq E$ with $I \cap O = \emptyset$, we find a maximum matching $M$ of $G$ such that $I \subseteq M$ and $O \cap M = \emptyset$ if it exists.
Such a maximum matching can be found by finding a matching of size $\nu(G) - \size{I}$ in the graph obtained from $G$ by deleting all end vertices of $I$ and edges in $O$.
Let $M \setminus I = \{e_1, e_2, \ldots, e_k\}$ and let $\mathcal M_i = \mathcal M(G, I \cup \{e_1, \ldots, e_{i-1}\}, O \cup \{e_i\})$ for $1 \le i \le k$.
Then, $\{\mathcal M_1, \mathcal M_2, \ldots, \mathcal M_k, \{M\}\}$ is a partition of $\mathcal M(G, I, O)$.
This partition process naturally yields a recursive enumeration algorithm, which defines a rooted tree whose node is labeled $(G, I, O)$ and outputs a maximum matching in $\mathcal M(G, I, O)$ at each node labeled $(G, I, O)$ in the rooted tree.
To upper bound the delay, we need to ensure that the algorithm outputs a maximum matching at each node in the rooted tree.
This can be done by checking whether $\mathcal M_i$ is nonempty for each $1 \le i \le k$.
Therefore, at each node, we additionally solve the maximum matching problem at most $k$ time and as $k \le n / 2$, the delay of the algorithm is upper bounded by $O(\min(n^{3/2}m, n^{\omega + 1}))$, where $\omega < 2.37$ is the exponent of matrix multiplication.

To further improve the running time of checking whether $\mathcal M_i = \mathcal M(G, I \cup \{e_1, \ldots, e_{i-1}\}, O \cup \{e_i\})$ is nonempty, we use the fact that $M$ is a maximum matching of $G$ with $I \subseteq M$ and $O \cap M = \emptyset$.
Let $I' = I \cup \{e_1, \ldots, e_{i-1}\}$ and $O' = O \cup \{e_i\}$.
Let $G'$ be the graph obtained from $G$ by deleting all the end vertices of edges in $I'$ and edges in $O'$.
Recall that it suffices to check whether $G'$ has a matching of cardinality $\nu(G) - \size{I'}$.
As $\size{M} = \nu(G)$, we have $\size{M \cap E(G')} = \size{M \setminus (I' \cup \{e_i\})}= \nu(G) - \size{I'} - 1$.
Moreover, since $M \cap E(G')$ is a matching of $G'$, it suffices to check whether $G'$ has an $(M \cap E(G'))$-augmenting path.
This can be done in time $O(m)$ \cite{DBLP:journals/computing/KamedaM74,DBLP:journals/jacm/Gabow76} and hence the delay is improved to $\order{nm}$.

\begin{theorem}\label{thm:nu=t}
    We can enumerate all maximum matchings of $G$ in delay $\order{nm}$ with polynomial space.    
\end{theorem}

It would be worth noting that the extension problem for enumerating {\em maximum} matching can be solved in polynomial time, whereas the {\em maximal} counterpart is intractable as proved in \Cref{theo:hard}.

\bibliographystyle{plain}
\bibliography{main.bib}

\begin{thebibliography}{10}

\bibitem{Avis::1996}
D.~Avis and K.~Fukuda.
\newblock Reverse search for enumeration.
\newblock {\em Discret. Appl. Math.}, 65(1):21 -- 46, 1996.

\bibitem{Birmele:2013}
Etienne Birmel{\'{e}}, Rui~A. Ferreira, Roberto Grossi, Andrea Marino, Nadia
  Pisanti, Romeo Rizzi, and Gustavo Sacomoto.
\newblock Optimal listing of cycles and st-paths in undirected graphs.
\newblock In {\em {Proc}. {SODA} 2013}, pages 1884--1896, 2013.

\bibitem{DBLP:conf/fct/CaselFGMS19}
Katrin Casel, Henning Fernau, Mehdi~Khosravian Ghadikolaei, J{\'{e}}r{\^{o}}me
  Monnot, and Florian Sikora.
\newblock Extension of some edge graph problems: Standard and parameterized
  complexity.
\newblock In {\em Proc. {FCT} 2019}, volume 11651 of {\em Lecture Notes in
  Computer Science}, pages 185--200. Springer, 2019.

\bibitem{DBLP:journals/dam/ChegireddyH87}
Chandra~R. Chegireddy and Horst~W. Hamacher.
\newblock Algorithms for finding k-best perfect matchings.
\newblock {\em Discret. Appl. Math.}, 18(2):155--165, 1987.

\bibitem{Cohen::2008}
Sara Cohen, Benny Kimelfeld, and Yehoshua Sagiv.
\newblock Generating all maximal induced subgraphs for hereditary and
  connected-hereditary graph properties.
\newblock {\em J. Comput. Syst. Sci.}, 74(7):1147--1159, 2008.

\bibitem{DBLP:journals/algorithmica/CominR18}
Carlo Comin and Romeo Rizzi.
\newblock An improved upper bound on maximal clique listing via rectangular
  fast matrix multiplication.
\newblock {\em Algorithmica}, 80(12):3525--3562, 2018.

\bibitem{Alessio:Roberto:ICALP:2016}
Alessio Conte, Roberto Grossi, Andrea Marino, and Luca Versari.
\newblock {Sublinear-Space Bounded-Delay Enumeration for Massive Network
  Analytics: Maximal Cliques}.
\newblock In {\em Proc. {ICALP} 2016}, volume~55 of {\em LIPIcs}, pages
  148:1--148:15. Schloss Dagstuhl--Leibniz-Zentrum fuer Informatik, 2016.

\bibitem{Conte::2019}
Alessio Conte and Takeaki Uno.
\newblock New polynomial delay bounds for maximal subgraph enumeration by
  proximity search.
\newblock In {\em Proc. {STOC} 2019}, pages 1179--1190, 2019.

\bibitem{Edmonds:paths:1965}
Jack Edmonds.
\newblock Paths, trees, and flowers.
\newblock {\em Canadian J. Math.}, 17:449–467, 1965.

\bibitem{DBLP:reference/algo/Eppstein16}
David Eppstein.
\newblock \emph{k}-best enumeration.
\newblock In {\em Encyclopedia of Algorithms}, pages 1003--1006. Springer, New
  York, NY, 2016.

\bibitem{Fellows:Kratochvil:Middendorf:Pfeiffer:Algorithmica:1995}
M~R Fellows, J~Kratochvil, M~Middendorf, and F~Pfeiffer.
\newblock The complexity of induced minors and related problems.
\newblock {\em Algorithmica}, 13(3):266--282, March 1995.

\bibitem{Fukuda:1992}
Komei Fukuda and Tomomi Matsui.
\newblock Finding all minimum-cost perfect matchings in bipartite graphs.
\newblock {\em Networks}, 22(5):461--468, 1992.

\bibitem{DBLP:journals/jacm/Gabow76}
Harold~N. Gabow.
\newblock {An Efficient Implementation of Edmonds' Algorithm for Maximum
  Matching on Graphs}.
\newblock {\em J. {ACM}}, 23(2):221--234, 1976.

\bibitem{Johnson:Yannakakis:Papadimitriou:IPL:1988}
David~S. Johnson, Mihalis Yannakakis, and Christos~H. Papadimitriou.
\newblock On generating all maximal independent sets.
\newblock {\em Inf. Process. Lett.}, 27(3):119 -- 123, 1988.

\bibitem{DBLP:journals/computing/KamedaM74}
T.~Kameda and J.~Ian Munro.
\newblock A {$O(|V|*|E|)$} algorithm for maximum matching of graphs.
\newblock {\em Computing}, 12(1):91--98, 1974.

\bibitem{Khachiyan::2006:matroid}
L.~Khachiyan, E.~Boros, K.~Borys, K.~Elbassioni, V.~Gurvich, and K.~Makino.
\newblock Enumerating spanning and connected subsets in graphs and matroids.
\newblock In {\em {Proc.} {ESA} 2006}, pages 444--455, Berlin, Heidelberg,
  2006. Springer Berlin Heidelberg.

\bibitem{Khachiyan::2008}
L.~Khachiyan, E.~Boros, K.~Borys, K.~Elbassioni, V.~Gurvich, and K.~Makino.
\newblock Generating cut conjunctions in graphs and related problems.
\newblock {\em Algorithmica}, 51(3):239–263, 2008.

\bibitem{Kobayashi:Efficient:2020}
Yasuaki Kobayashi, Kazuhiro Kurita, and Kunihiro Wasa.
\newblock Efficient constant-factor approximate enumeration of minimal subsets
  for monotone properties with weight constraints.
\newblock {\em CoRR}, abs/2009.08830, 2020.

\bibitem{DBLP:journals/corr/abs-2012-09153}
Tuukka Korhonen.
\newblock Listing small minimal separators of a graph.
\newblock {\em CoRR}, abs/2012.09153, 2020.

\bibitem{DBLP:conf/mfcs/KuritaK20}
Kazuhiro Kurita and Yasuaki Kobayashi.
\newblock Efficient enumerations for minimal multicuts and multiway cuts.
\newblock In {\em {Proc.} {MFCS} 2020}, pages 60:1--60:14, 2020.

\bibitem{Lawler1972}
Eugene~L. Lawler.
\newblock {A Procedure for Computing the K Best Solutions to Discrete
  Optimization Problems and Its Application to the Shortest Path Problem}.
\newblock {\em Manage. Sci.}, 18(7):401--405, 1972.

\bibitem{DBLP:conf/swat/MakinoU04}
Kazuhisa Makino and Takeaki Uno.
\newblock New algorithms for enumerating all maximal cliques.
\newblock In {\em {Proc.} {SWAT} 2004}, volume 3111 of {\em Lecture Notes in
  Computer Science}, pages 260--272. Springer, 2004.

\bibitem{DBLP:conf/focs/MicaliV80}
Silvio Micali and Vijay~V. Vazirani.
\newblock {An $O(\sqrt{|V|}|E|)$ Algorithm for Finding Maximum Matching in
  General Graphs}.
\newblock In {\em Proc. {FOCS} 1980}, pages 17--27, 1980.

\bibitem{Morrison:PATRICIA:1968}
Donald~R. Morrison.
\newblock Patricia—practical algorithm to retrieve information coded in
  alphanumeric.
\newblock {\em J. ACM}, 15(4):514–534, 1968.

\bibitem{DBLP:conf/focs/MuchaS04}
Marcin Mucha and Piotr Sankowski.
\newblock {Maximum Matchings via Gaussian Elimination}.
\newblock In {\em {Proc.} {FOCS} 2004}, pages 248--255. {IEEE} Computer
  Society, 2004.

\bibitem{Murty:Letter:1968}
Katta~G. Murty.
\newblock {Letter to the Editor—An Algorithm for Ranking all the Assignments
  in Order of Increasing Cost}.
\newblock {\em Oper. Res.}, 16(3):682--687, 1968.

\bibitem{Schwikowski::2002}
Benno Schwikowski and Ewald Speckenmeyer.
\newblock On enumerating all minimal solutions of feedback problems.
\newblock {\em Discret. Appl. Math.}, 117(1-3):253--265, 2002.

\bibitem{Takata::2010}
K.~Takata.
\newblock Space-optimal, backtracking algorithms to list the minimal vertex
  separators of a graph.
\newblock {\em Discrete Applied Mathematics}, 158(15):1660 -- 1667, 2010.

\bibitem{DBLP:journals/siamcomp/TsukiyamaIAS77}
Shuji Tsukiyama, Mikio Ide, Hiromu Ariyoshi, and Isao Shirakawa.
\newblock {A New Algorithm for Generating All the Maximal Independent Sets}.
\newblock {\em {SIAM} J. Comput.}, 6(3):505--517, 1977.

\bibitem{DBLP:conf/isaac/Uno97}
Takeaki Uno.
\newblock Algorithms for enumerating all perfect, maximum and maximal matchings
  in bipartite graphs.
\newblock In {\em Proc. {ISAAC}'97}, pages 92--101, 1997.

\bibitem{Uno:NIIjournal}
Takeaki Uno.
\newblock A fast algorithm for enumerating non-bipartite maximal matchings.
\newblock {\em NII Journal}, 3:89--97, 2001.

\bibitem{Uno:2015}
Takeaki Uno.
\newblock Constant time enumeration by amortization.
\newblock In {\em {Proc.} {WADS} 2015}, pages 593--605, Cham, 2015. Springer
  International Publishing.

\end{thebibliography}
\end{document}